\documentclass[runningheads,orivec]{llncs}
\usepackage[T1]{fontenc}

\usepackage[noEnd,commentColor=black]{algpseudocodex}
\usepackage{amsmath}
\usepackage{amssymb}

\usepackage{amsthm}
\usepackage{mathtools}
\usepackage[mathscr]{euscript}
\usepackage[shortlabels]{enumitem}
\usepackage{graphicx} %
\usepackage{thmtools}
\usepackage{todonotes}
\usepackage{xcolor}
\usepackage{microtype}
\usepackage{comment}

\usepackage{thm-restate}
\usepackage[T1]{fontenc}
\usepackage{lmodern}
\usepackage[utf8]{inputenc}

\usepackage[normalem]{ulem}

\usepackage{tcolorbox}

\usepackage{upgreek}

\usepackage{authblk}
\usepackage[normalem]{ulem}

\usepackage{tikz}
\usetikzlibrary{calc}

\newcommand{\R}{\mathbb{R}}
\newcommand{\fO}{\mathcal{O}}

\newcommand{\cR}{\mathcal{R}}

\newcommand{\eat}[1]{}

\DeclareMathOperator{\supp}{supp} %

\newcommand{\update}{\mathsf{update}}

\newcommand{\cP}{\mathcal{P}}
\newcommand{\cH}{\mathcal{H}}

\renewcommand{\cR}{\mathcal{R}}
\newcommand{\cQ}{\mathcal{Q}}
\newcommand{\cC}{\mathcal{C}}

\newcommand{\cS}{\mathcal{S}}

\newcommand{\connected}[1]{\def\temp{#1}\ifx\temp\empty\sim\else\overset{#1}{\sim}\fi}

\tikzset{
	point/.style={circle, fill, inner sep=1.5pt},
	smallpoint/.style={point, inner sep=1.2pt},
	tinypoint/.style={point, inner sep=1pt},
	hlbox/.style={fill, {white!90!black}},
	subrect/.style={draw, fill={white!80!cyan}}, %
	msrect/.style=subrect %
}

\usepackage{newtxtext}       %
\usepackage[varvw]{newtxmath}       %

\usepackage{graphicx}

\usepackage[colorlinks]{hyperref}
\usepackage[capitalize]{cleveref}

\crefname{equation}{eq.}{eqs.} %
\crefname{enumi}{}{} %
\crefname{icase}{case}{cases}
\crefname{ipart}{part}{parts}
\crefname{iprop}{property}{properties}
\crefname{iinv}{invariant}{invariants}

\crefformat{section}{\S\,#2#1#3}
\crefformat{subsection}{\S\,#2#1#3}
\crefrangeformat{section}{\S\S\,#3#1#4--#5#2#6}
\crefmultiformat{section}{\S\S\,#2#1#3}{,~#2#1#3}{,~#2#1#3}{,~#2#1#3}

\begin{document}
\title{Fast and simple multiplication of bounded twin-width matrices}
\author{L\'aszl\'o Kozma\inst{1} \and
Jongseo Lee\inst{2} \and
Michal Opler\inst{3}}
\authorrunning{L. Kozma, J. Lee, M. Opler}
\institute{Faculty of Computer Science, TU Dresden, Germany \and
School of Computing, KAIST, South Korea \and
Czech Technical University in Prague, Czech Republic}
\maketitle              %

\newcommand{\JL}[1]{{\color{blue} #1}}

\begin{abstract}

Twin-width is a natural complexity measure of matrices (and more general structures) that has recently emerged as a unifying concept with important algorithmic applications. While the twin-width of a matrix is invariant to reordering rows and columns, most of its algorithmic applications to date assume that the input is given in a certain canonical ordering that yields a bounded twin-width contraction sequence. In general, efficiently finding such a sequence -- even for an approximate twin-width value -- remains a central and elusive open question.

In this paper we show that a binary $n \times n$ matrix of twin-width $d$   can be preprocessed in $\widetilde{\fO}_d(n^2)$ time, so that its product with any vector can be computed in $\widetilde{\fO}_d(n)$ time. Notably, the twin-width of the input matrix need not be known and no particular ordering of its rows and columns is assumed. If a canonical ordering is available, i.e., if the input matrix is $d$-twin-ordered, then the runtime of preprocessing and matrix-vector products can be further reduced to $\fO(n^2)$ and $\fO(dn)$.

Consequently, we can multiply two $n \times n$ matrices in $\widetilde{\fO}(n^2)$ time, when \emph{at least one} of the matrices consists of 0/1 entries and has bounded twin-width. The results also extend to the case of bounded twin-width matrices with adversarial corruption. Our algorithms are significantly faster and simpler than earlier methods that involved first-order model checking and required both input matrices to be $d$-twin-ordered.

\end{abstract}

\section{Introduction}\label{sec1}

  Matrix multiplication is a fundamental operation in all of computer science, and efficient algorithms for matrix multiplication have been intensively researched since the beginnings of the field. Assuming unit cost scalar operations, the na\"{i}ve way to multiply two $n \times n$ matrices takes $\fO(n^3)$ time. This was first improved by Strassen~\cite{Strassen} to $\fO(n^{2.81})$, following Karatsuba's fast integer multiplication algorithm~\cite{karatsuba1962multiplication}. Since then, the complexity of matrix multiplication has been shown to be in $\fO(n^\omega)$, with $\omega < 2.37134$~\cite{AlmanDWXXZ25}; further improving this bound remains one of the central questions of computer science; see~\cite{Blaser13,burgisser2013algebraic} for broad surveys.

Applications of matrix multiplication arise in virtually all computational fields. A closely related task is to preprocess an $n \times n$ matrix $M$, such as to allow efficiently computing matrix-vector products $Mv$ for arbitrary size-$n$ vectors $v$ presented online. An algorithm for the latter task that takes time $\mathscr{P}$ for preprocessing $M$ and time $\mathscr{T}$ for computing $Mv$ clearly also implies an algorithm for multiplying two $n \times n$ matrices in time $\mathscr{P}+\mathscr{T}\cdot n$. Matrix-vector products (often in a setting where the same matrix is repeatedly used and can thus be preprocessed) are a core subroutine and efficiency-bottleneck in many areas, including numerical optimization, e.g., the conjugate gradient method~\cite{golub2013matrix}, machine learning (both at the training and at the inference stage)~\cite{sze2017efficient,bengio2017deep,shalev2014understanding}, finite element analysis~\cite{zienkiewicz2005finite}, computer graphics~\cite{shirley2009fundamentals}, and graph algorithms~\cite{vishnoi2010laplacian,spielman2019spectral}, to mention just a few. For fully general input, almost-quadratic lower bounds are known for the complexity of the problem in various models of computation~\cite{CliffordGL15,FrandsenHM01,ChakrabortyKL18,HenzingerLS22}. The OMv conjecture~\cite{henzinger2015unifying} implies that online matrix-vector products cannot be computed in $\fO(n^{2-\varepsilon})$ time (for any $\varepsilon>0$), even if a polynomial-time preprocessing of the matrix is allowed, and even in the Boolean multiplication case. This conjecture has been the basis of a multitude of conditional hardness results for dynamic graph-, string-, and other algorithmic problems, e.g., see~\cite{hanauer2022recent,jin2022tight,jiang2023complexity,van2019dynamic,clifford2018,HenzingerLS22}.

Given the prevalence of the online matrix-vector product and its worst-case computational hardness, it is natural to aim for faster methods on inputs with some intrinsic structure, beyond the notions of \emph{sparsity} or \emph{low rank} that are classically exploited, e.g., see~\cite{golub2013matrix}. In this direction, in very recent work, Anand, van den Brand, and McCarty~\cite{anand2025structural} consider input matrices of low Vapnik-Chervonenkis (VC) dimension. For $n \times n$ matrices of VC-dimension $d$ they show that matrix-vector products can be computed in time $\widetilde{\fO}(n^{2-1/d})$ after $\widetilde{\fO}(n^{2})$ preprocessing.\footnote{The notation $\widetilde{\fO}(\cdot)$ suppresses a factor $\log^c{n}$, for some $c>0$.} VC-dimension plays an important role in learning theory~\cite{shalev2014understanding} and in computational geometry~\cite{matousek2013lectures,har2011geometric}, but imposes a rather mild restriction on the structure of the input. Accordingly, the above runtime for matrix-vector multiplication is -- already for moderate values of the VC-dimension -- close to quadratic.

In this paper, we study the complexity of matrix-vector multiplication under a stronger structural restriction: \emph{bounded twin-width}.\footnote{One way to make the comparison quantitative is that the growth rate of binary $n \times n$ matrices with bounded VC-dimension is $2^{\fO(n^{2-\varepsilon})}$~\cite{alon2011structure} whereas the growth rates of bounded twin-width and bounded twin-order matrices are respectively $2^{\fO(n\log{n})}$ and $2^{\fO(n)}$~\cite{bonnet2024twin4}.} In this regime we obtain close to linear runtimes for matrix-vector-multiplication, after close to quadratic-time preprocessing (i.e., both runtimes are essentially best possible).

Twin-width is a recently introduced complexity measure of matrices (as well as of other combinatorial structures) that has emerged as a central and unifying concept~\cite{bonnet2021twin}. Bounded twin-width generalizes bounded tree-width, clique-width, and rank-width, and in graphs it subsumes the minor-free property. Matrices of bounded twin-width are, in general, neither sparse, nor low-rank. In \emph{ordered graphs}, bounded twin-width, in some sense, captures the boundary of tractability for first-order model checking~\cite{bonnet2024twin4,thomasse2022brief}.

Informally, bounded twin-width matrices admit a sequence of merges (``contractions'') of pairs of rows or columns, preserving a certain local homogeneity, i.e., ensuring that merged rows and columns do not significantly differ (we give precise definitions later in the section).
Several algorithmic applications of twin-width were found in the past few years, where otherwise hard problems become %
tractable on bounded twin-width input~\cite{twin_kernel,kratsch2022triangle,matrix_ds,bonnet2024twin3}. These, however, mostly assume that the input matrix (or graph) is ordered, i.e., its rows and columns are given in a canonical order that yields a contraction sequence in the above sense, or that a tree of contractions is explicitly given.

In the unordered case, the general algorithmic usefulness of twin-width has so far been less evident. The existence of a polynomial-time or FPT algorithm that approximates twin-width (even to a factor that can depend on the twin-width value itself) remains an elusive open question~\cite{bonnet2021twin}. In many important special cases (e.g., minor-free graphs), an ordering can be efficiently found, from which a bounded twin-width contraction sequence can be reconstructed. Even in these cases, going from an intermediate ordering to a contraction sequence is highly non-trivial and involves a super-exponential blow-up of the twin-width parameter, limiting the usefulness of the technique in the context of matrix multiplication. Recognition of matrices of twin-width $4$ is already NP-hard, as is approximating twin-width with a factor better than $5/4$~\cite{twin_npc}.

 Our main result sidesteps these complexity issues by exploiting an alternative structural consequence of bounded twin-width that we prove. Given this characterization, our algorithms benefit from low twin-width without requiring an explicit contraction scheme that certifies it.

Matrix multiplication with bounded twin-width inputs has been studied in~\cite{bonnet2023twin5}. They give an algorithm that multiplies two $n \times n$ matrices of twin-width at most $d$ in time $\widetilde{\fO}_d(n^2)$.
We improve upon this result in multiple aspects: (i) we only require one of the matrices to have bounded twin-width, while in~\cite{bonnet2023twin5} both matrices obey this restriction, (ii) our methods are significantly simpler, compared to the first-order logic based technique developed in~\cite{bonnet2023twin5}, with small constants and mild dependence on $d$, bringing the approach into the realm of practicality, and most importantly: (iii) our algorithms require neither a canonical ordering of the matrices nor a contraction sequence; %
the bound on the twin-width is only used in the analysis. If the matrix is, in fact, given in an ordering that certifies bounded twin-width, then our algorithms are even faster and simpler.\footnote{Related results appeared concurrently with and after the preliminary version~\cite{twv1} of this paper~\cite{bonnet26,bonnet2026b,CMY26} and (i) and (iii) also hold for them; we further discuss these results, and the connection with~\cite{anand2025structural} at the end of the section.}

We now proceed with the main definitions and more precise statement of results.

\smallskip

\noindent\emph{Twin-width and twin-order.} Consider a binary matrix $M$. Following~\cite{matrix_ds}, a \emph{merge sequence} for $M$ is a sequence of {merges} between neighboring pairs of rows or columns, until a single entry remains. A \emph{merge} of two columns replaces them with a single column where the $i$-th entry of the new column is $0$ or $1$, if both original columns had a value of $0$ (resp.\ $1$) at position $i$. Otherwise, the $i$-th entry of the new column is the ``error'' symbol $\bot$; this can happen if the $i$-th entries of the two merged columns differ or if at least one of them is the symbol $\bot$. Merges of rows are analogously defined.

Matrix $M$ is \emph{$d$-twin-ordered}, if there is a merge sequence for $M$ in which every row and column of every intermediate matrix has at most $d$ symbols $\bot$.
The \emph{twin-width} of~$M$ is the smallest value $d$ for which the rows and columns can be reordered so that the resulting matrix is $d$-twin-ordered. %
Note that this reordering implies a significant shift in complexity. For instance, an $n \times n$ matrix with a chessboard-like arrangement of $0$s and $1$s is not $d$-twin-ordered for any $d<n$, but its twin-width is two. %

\smallskip

\noindent\emph{Our results.} Our first result concerns the easier case of $d$-twin-ordered matrices. We give a simple and efficient multiplication algorithm by combining: (i) a known structural characterization of $d$-twin-ordered matrices in terms of decomposition into rectangles, (ii) a reformulation of this property as the sparsity of the \emph{discrete derivative} of the matrix, and (iii) an efficient multiplication algorithm for matrices with sparse derivative. %

A \emph{$1$-rectangle} in a binary matrix is a contiguous submatrix (the intersection of a set of consecutive rows and a set of consecutive columns) whose entries are all $1$. A \emph{rectangle-decomposition} of $M$ is a collection of disjoint $1$-rectangles that together cover all $1$-entries of $M$.
The following structural result is well-known~\cite{bonnet2024twin3}, see~\cite{matrix_ds} for a self-contained proof in the context of matrices.

\begin{lemma}[\cite{bonnet2024twin3}]\label{lem1}
Let $M \in \{0,1\}^{n \times n}$ be a $d$-twin-ordered matrix. Then $M$ admits a rectangle-decomposition $\cR$ with $|\cR| \leq d(2n-2)+1 \in \fO(dn)$.
\end{lemma}

The next ingredient of our algorithm is the \emph{discrete derivative} of a matrix.
The \emph{discrete derivative} (or two-dimensional prefix difference array) $D$ of a matrix $M$ is an $(n+1)\times (n+1)$ matrix defined by
\begin{equation}
D_{i, j} := M_{i, j} - M_{i-1, j} - M_{i, j-1} + M_{i-1, j-1}, \label{eqn1}
\end{equation}
where we let $M_{i,j} = 0$ whenever $i$ or $j$ are not in $[n]$.
Then, it is easy to observe that $M_{i, j} = \sum_{a \le i} \sum_{b \le j} D_{a, b}$.
Define the \emph{support} $\supp(A)$ of a matrix $A$ as the set of non-zero entries of $A$. The following lemma links the two matrix-parameters introduced so far.

\begin{restatable}{lemma}{lema}
\label{lem2}

Let $M \in \{0,1\}^{n \times n}$ be a binary matrix admitting a rectangle-decomposition with at most $k$ rectangles.
Then, the discrete derivative $D$ of $M$ satisfies $\lvert \supp(D)\rvert \le 4k$.

\end{restatable}

The final ingredient is an algorithm that efficiently multiplies a matrix with a vector using the discrete derivative.

\begin{restatable}{lemma}{lemaa}\label{lem3}
Let $M \in \{0, 1 \}^{n \times n}$ be a binary matrix and $D$ be its discrete derivative.
Then, given the list of non-zero entries of $D$, there is an $\fO(n + \lvert \supp(D) \rvert)$ time algorithm computing the matrix-vector product $Mv$ for an arbitrary column vector $v \in \mathbb{R}^n$.
\end{restatable}

Combining the three lemmas, we obtain the following theorem.\footnote{In an  \href{https://arxiv.org/abs/2602.20023v1}{earlier version} of our paper~\cite{twv1} we computed an explicit, near-optimal rectangle-decomposition and used the rectangles for efficient matrix-multiplication. The use of the discrete derivative $D$ further simplifies the algorithm; a rectangle-decomposition is now needed only in the analysis. Note that $D$ can be easily computed from $M$ via \eqref{eqn1}, in $\fO(n^2)$ time.}

\begin{restatable}{theorem}{theom}
\label{thm1}
Let $M \in \{0,1\}^{n \times n}$ be a $d$-twin-ordered binary matrix. We can preprocess $M$ in time $\fO(n^2)$ so that for arbitrary column vectors $v \in \mathbb{R}^n$, we can compute the products $M v$ or $v^\intercal M$ in time $\fO(dn)$.
\end{restatable}

Observe that the result implies that multiplying two $n \times n$ matrices $A$ and $B$ takes time ${\fO}(dn^2)$, whenever at least one of the matrices is $d$-twin-ordered. As discussed, this significantly simplifies the results of~\cite{bonnet2023twin5} that involved extra logarithmic factors and an exponential dependence on $d$, while requiring both matrices to be $d$-twin-ordered.

\medskip

The general bounded twin-width case appears more challenging, since the ordering of rows and columns does not immediately yield an efficient contraction sequence or a small rectangle-decomposition or a sparse derivative in the sense of Lemmas~\ref{lem1} and \ref{lem2}.

Our main theorem achieves the following.

\begin{theorem}\label{thm2}
Let $M \in \{0,1\}^{n \times n}$ be a binary matrix with twin-width $d$. We can preprocess $M$ in time $\widetilde{\fO}_d(n^2)$ so that for arbitrary column vectors $v \in \R^n$, we can compute the products $M v$ or $v^\intercal M$ in time $\widetilde{\fO}_d(n)$.
\end{theorem}

We again emphasize that in contrast to earlier twin-width-based algorithms, the result does not require a contraction-sequence to be provided as part of the input, and no specific ordering of the rows and columns of the input matrix is assumed.

The main structural property we use relies on Hamming distances between pairs of rows and columns of a matrix. The Hamming distance $d(u,v)$ between two binary vectors $u,v \in \{0,1\}^n$ is the number of entries in which they differ, i.e., $d(u,v) = \sum_{i=1}^n{|u_i - v_i|}$.

Let $r^{(1)}, \dots, r^{(n)}$ denote the rows of matrix $M \in \{0,1\}^{n\times n}$. For a given permutation~$\pi$ of the indices $\{1,\dots,n\}$, we denote by $\cH_\pi^r(M)$ the \emph{row-Hamming-sum} of $M$, the sum of pairwise Hamming-distances between neighboring rows in the order given by $\pi$.
Formally, $\cH_\pi^r(M) = \sum_{i=1}^{n-1}{d(r^{(\pi(i))},r^{(\pi(i+1))})}$.

The \emph{row-Hamming-coherence} of $M$, denoted $\cH^r(M)$, is the row-Hamming-sum of $M$ for the ordering~$\pi$ of rows that minimizes this quantity, i.e., $\cH^r(M) = \min_\pi\{\cH_\pi^r(M)\}$.
The \emph{column-Hamming-sum} $\cH_\pi^c(M)$ and \emph{column-Hamming-coherence} $\cH^c(M)$ are defined entirely analogously.

We now rely on a known algorithm for matrix-vector multiplication with preprocessing, when the matrix has low Hamming-coherence. Such an algorithm was described by Bj\"{o}rklund and Lingas~\cite{BjorklundLingas} and extended by Gąsieniec and Lingas~\cite{GasieniecLingas}.

\begin{lemma}[\cite{BjorklundLingas,GasieniecLingas}]
\label{lem7}
Let $M \in \{0,1\}^{n \times n}$ be a binary input matrix with a row-Hamming-coherence $\cH^r$ . Then $M$ can be preprocessed in time $\widetilde{\fO}(n^2)$, so that for any vector $v$, the product $Mv$ can be computed in time $\widetilde{\fO}(n + \cH^r)$.
\end{lemma}

Assuming that the ordering $\pi$ realizing the optimal row-Hamming-sum is known, multiplication works simply by representing successive rows in terms of their difference from the previous row. Then, the contribution of each row to the result can be computed from the contribution of the previous row in the ordering, adjusted with the effect of the changed elements. Since the effect of each change can be added in constant time, the total time is proportional to the total row-Hamming-sum $\cH_\pi^r(M)$. We note that low Hamming-sum orderings were similarly used for matrix multiplication in~\cite{anand2025structural}.

Note, however, that the optimal ordering $\pi$ is not available. Let us sketch how we can obtain an approximation of it, during preprocessing. %
(Finding the optimum is intractable  as it amounts to solving a metric TSP instance.) An ordering that yields at most twice the optimal row-Hamming-sum can be obtained by the well-known approximation based on a minimum spanning tree (MST). Computing an MST of the row vectors under Hamming distance would na\"{i}vely require $\fO(n^3)$ time, which is not affordable for preprocessing. Utilizing fast approximate nearest neighbor data structures (e.g., based on random projections), the task can be solved in $\widetilde{\fO}(n^2)$ time, at the cost of a further $\fO(\log{n})$ factor loss in approximating $\cH^r(M)$. We thus obtain an ordering that yields a row-Hamming-sum of $\fO(\log{n}) \cdot \cH^r$. Since this is the ordering we use for matrix-vector products, an additional logarithmic factor also appears in the runtime of that operation.

We use both the preprocessing and multiplication part of this algorithm as a black box, and refer to~\cite{BjorklundLingas,GasieniecLingas,anand2025structural} for details. An analogous result for the product $v^{\intercal}M$ clearly holds by replacing row-Hamming-coherence $\cH^r$ by column-Hamming-coherence $\cH^c$.

It remains to connect the twin-width and Hamming-coherence measures of matrices. Our result, proved in \S\,\ref{sec4}, is the following.

\begin{restatable}{theorem}{thma}
\label{thm3}
Let $M \in \{0,1\}^{n \times n}$ be a binary matrix of twin-width $d$. Then both the row-Hamming-coherence and column-Hamming-coherence of $M$ can be bounded as $\cH^r(M), \cH^c(M) \in 2^{2^{\fO(d)}} \cdot n \log{n}$.
\end{restatable}

Theorem~\ref{thm2} then immediately follows from Theorem~\ref{thm3} via Lemma~\ref{lem7}.
Crucially, given a matrix of bounded twin-width, the bound on the Hamming-coherences (Theorem~\ref{thm3}) holds regardless of the ordering of rows and columns. To see this, observe that, for the row-Hamming-sum, a reordering of the rows is anyway undone by the optimal choice of~$\pi$, and the ordering of columns does not affect the Hamming sums, by the commutativity of addition in the definition of Hamming distance.

It may be tempting to conjecture that the reverse implication also holds, i.e., that low Hamming-coherence implies low twin-width. This is, however, not the case: there exist (sub)cubic graphs with unbounded twin-width~\cite{BonnetGKTW22}, while both the row- and column-Hamming-sums of their adjacency matrix (in any ordering) are clearly linear.

In terms of $n$, we show our bound to be best possible. %

\begin{restatable}{theorem}{thmb}
\label{thm4}
For every $n$, there exists a binary matrix $M \in \{0,1\}^{n \times n}$ with twin-width $d \leq 3$ and row- and column-Hamming-coherence $\cH^r(M), \cH^c(M) \in \Omega(n \log{n})$.
\end{restatable}

\medskip

We have shown how to efficiently preprocess a matrix of bounded twin-width, respectively bounded twin-order, such as to allow efficient matrix-vector multiplications.
The two cases sit at two extremes. Typically, in past applications of twin-width, bounded twin-width matrices are first brought into a canonical ordering of both rows and columns called \emph{$d$-mixed-free} ordering, that is less constrained than $d$-twin-order. It is known that every $d$-twin-ordered matrix is also $(2d+2)$-mixed-free~\cite{bonnet2021twin}, but a matrix with bounded mixed-freeness may be of unbounded twin-order (in the same fixed ordering). We give the precise definition in \S\,\ref{sec3}.

For a matrix in $d$-mixed-free ordering, a bounded twin-width merge sequence can be found, i.e., the matrix can be brought into an $f(d)$-twin-ordered form. Existing algorithms for this task, however, result in a blow-up in complexity according to $f(d) = 2^{2^{\fO(d)}}$.

We show that in the case of matrix multiplication this transformation can be avoided and the dependence on the complexity-parameter can be reduced from doubly- to singly-exponential, using a simpler characterization of $d$-mixed-free matrices.

\begin{restatable}{lemma}{lemb}
\label{lem6}
Let $M \in \{0,1\}^{n \times n}$ be a $d$-mixed-free binary matrix. Then $M$ admits a rectangle-decomposition $\cR$ with $|\cR| \in 2^{\fO(d)} \cdot n$.
\end{restatable}

Lemma~\ref{lem6} can be seen as analogous to Lemma~\ref{lem1}, and we can similarly employ Lemmas~\ref{lem2} and \ref{lem3} to obtain matrix-vector multiplication in time $\fO_d(n)$ after preprocessing of time $\fO_d(n^2)$, in case of $d$-mixed-free matrices.

We show that the exponential dependence on the parameter~$d$ is unavoidable through this approach since there exist $d$-mixed-free matrices which require $2^{\Omega(d^{1/2})} \cdot n$ rectangles to cover their 1-entries.

\begin{restatable}{lemma}{lemc}
\label{lem8}
For every $d$ and $n$, there exists a $d$-mixed-free matrix $M \in \{0,1\}^{n \times n}$ such that every rectangle-decomposition of $M$ has size $2^{\Omega(d^{1/2})} \cdot n$.
\end{restatable}

Finally, observe that all our results have a certain robustness to corruption of the input.
For $d$-twin-ordered and $d$-mixed-free matrices (Theorem~\ref{thm1} and its analogue for $d$-mixed-free matrices), the runtimes depend mainly on the support-size of the discrete derivative of the input.
Flipping a single element of the matrix from $0$ to $1$ or from $1$ to $0$ clearly adds at most a constant number of entries to the support. Thus, corrupting $r$ elements only adds a term $\fO(r)$ to the runtimes of matrix-vector products. Perturbing a linear number of elements is therefore simply absorbed in the asymptotic cost, while perturbing $o(n^2)$ elements still preserves a non-trivial improvement to the quadratic runtime of matrix-vector multiplication, without increasing the preprocessing cost. Note that such a perturbation can maximally increase twin-width, thus the results would not follow simply from the twin-width bounds.

In case of the general bounded twin-width result (Theorem~\ref{thm2}), handling corruptions is similarly simple. Flipping an element can increase the Hamming distance between two rows or columns by at most one; we can thus similarly accommodate perturbing $o(n^2)$ elements while preserving $o(n^2)$-time matrix-vector multiplication.

We note that the discrete derivative representation of a matrix (used in Theorem~\ref{thm1}) can also accommodate efficient dynamic updates (changes of elements in $M$), as only a constant number of entries of $D$ need to be updated.

\smallskip

\noindent\textit{Structure of the paper.}
In \S\,\ref{sec2} we discuss $d$-twin-ordered matrices, proving Lemmas~\ref{lem2} and \ref{lem3}, and thereby Theorem~\ref{thm1}.
 In \S\,\ref{sec4} we study general bounded twin-width matrices, proving Theorem~\ref{thm3}, and thereby Theorem~\ref{thm2}, as well as Theorem~\ref{thm4}. The study of $d$-mixed-free matrices, including Lemmas~\ref{lem6} and \ref{lem8} is in \S\,\ref{sec3}. We briefly conclude in \S\,\ref{sec5}.

\smallskip
\noindent\textit{Concurrent and subsequent work.}
A preliminary version of this paper~\cite{twv1} appeared on arXiv a few days after, and independently of, the work of Bonnet, Geniet, Kim, and Moon~\cite{bonnet26}.
Among other results, they give an algorithm to multiply the adjacency matrix of a bounded twin-width graph with an
arbitrary $n \times n$ matrix in $\mathcal{O}(n^2 \log n)$ time with high probability, also
without an ordering or contraction sequence. Their approach, via signed tree models and interval biclique partitions, is more involved, but more general.

Bonnet et al.~\cite{bonnet26} also note that the analysis of~\cite{anand2025structural} extends from
VC-dimension to VC-density. Since bounded twin-width implies linear neighborhood
complexity~\cite{bonnet2024neighbourhood}, and hence VC-density one, this also gives an $\fO_d(n \log n)$ bound on Hamming-coherence, and thus an alternative route to
$\widetilde{\mathcal{O}}(n)$-time matrix-vector products. Our Theorem~\ref{thm3} proves such a bound directly from merge sequences, and by Theorem~\ref{thm4} the bound on the Hamming-coherence is tight.

Subsequently, Cardinal, McCarty, and Yuditsky~\cite{CMY26} obtained $\widetilde{\mathcal{O}}(n)$-time matrix-vector multiplication for all classes of linear neighborhood complexity, and Bonnet, Duron, Pilipczuk, Soko{\l}owski, and Toru\'nczyk~\cite{bonnet2026b} obtained the optimal bounds for such classes, $\mathcal{O}(n^2)$ preprocessing and $\mathcal{O}(n)$-time products, based on interval biclique partitions of linear size.
While these results are more general, their dependence on $d$ is not explicit; our approach is arguably simpler, and yields, for ordered input, an $\mathcal{O}(dn)$-time multiplication, with explicit linear dependence on $d$.

\section{Algorithm for $d$-twin-ordered matrices}
\label{sec2}

In this section, we describe our algorithm for matrix-vector multiplication when the $n \times n$ input matrix $M$ consists of elements $0$ and $1$ and is $d$-twin-ordered, thereby proving Theorem~\ref{thm1}.
We recall the two lemmas that imply the guarantees needed for the efficient algorithm, proceeding with their proofs.

\lema*

\lemaa*

\begin{proof}[Proof of \Cref{lem2}]
Recall that a rectangle-decomposition $\cR$ is a collection of pairwise disjoint rectangles $R_1, \dots, R_{|\cR|}$ that together cover all $1$-entries of $M$.
Formally, $R_i \cap R_j = \emptyset$ for all $R_i, R_j \in \cR$, and $M_{x,y}=1$ if and only if $(x,y) \in \bigcup_{R \in \cR} R$.

Observe that if $M = M_1 + M_2$ for two $n \times n$ binary matrices $M_1$ and $M_2$ with discrete derivatives $D_1$ and $D_2$, then the discrete derivative of~$M$ is equal to $D_1 + D_2$.
Furthermore, the discrete derivative of a matrix with all 1-entries in a single rectangle
$R = [r_1, r_2] \times [c_1, c_2]$ has non-zero entries exactly at the four positions
$(r_1, c_1)$, $(r_1, c_2+1)$, $(r_2+1, c_1)$, and $(r_2+1, c_2+1)$.

Therefore, if $M$ is the disjoint union of $k$ rectangles, we can express it as a sum $M = \sum_{i=1}^k M_i$ where each $M_i$ contains 1-entries only in the rectangle $R_i$, and thus we get $\lvert \supp(D) \rvert \le 4k$.
\end{proof}

\begin{proof}[Proof of \Cref{lem3}]
Let $\mathcal{S}_D := \{ (i, j, D_{i, j}) : (i, j) \in \supp(D) \}$ be the given list
of non-zero entries of $D$.
We describe the procedure for multiplying $M$ with a vector $v \in \mathbb{R}^n$.

First compute suffix sums of the vector $v$, i.e., $s_y := \sum_{j=y}^n v_j$, and $s_{n+1} := 0$, in $\fO(n)$ total time.
Then, for each row index $x$, compute $g_x := \sum_{y} D_{x, y} s_y$ by iterating over $\mathcal{S}_D$ (i.e., the support of $D$) in $\fO(\lvert \supp(D) \rvert)$ time.
Finally, take prefix sums of $g$ to compute
$(Mv)_i = \sum_{x = 1}^i g_x$.

  Correctness follows from the chain of equalities:
 \begin{multline*}
 (Mv)_i
= \sum_y M_{i,y} v_y
= \sum_j{v_j} \sum_{y \le j} (M_{i,y}-M_{i,y-1})
= \sum_y (M_{i,y}-M_{i,y-1})s_y=\\
= \sum_y \sum_{x\leq i} D_{x,y} s_y
= \sum_{x\leq i} \sum_{y} D_{x,y} s_y
= \sum_{x\leq i} g_x.\qedhere
 \end{multline*}
\end{proof}

The total running time for computing $Mv$ is $\fO(n + \lvert \supp(D) \rvert)$. For completeness, we show the pseudocode in \Cref{fig:fig1upd}.
Computing $v^\intercal M$ is symmetric; we omit further details.

The theorem now follows easily.

\begin{figure}[ht]
\small
    \textbf{Input:} ~~~Discrete derivative $\mathcal{S}_D$, column vector $v \in \mathbb{R}^n$\\
    \textbf{Output:} ~~$u \in \mathbb{R}^n$, where $u = Mv$
    \begin{algorithmic}
        \State $s_{n+1} \gets 0$
        \For{$y = n,n-1,\ldots,1$}
            \State $s_y \gets s_{y+1} + v_y$
            \quad \quad \textit{$\triangleright$ suffix sums of $v$}
        \EndFor
        \State $g \gets 0^{n+1}$
        \ForAll{$(x, y,D_{x, y}) \in \mathcal{S}_D$}
            \State $g_x \gets g_x + D_{x, y} s_y$
        \EndFor
        \State $u_0 \gets 0$ \quad \quad  \quad \quad  \quad ~~ \quad \textit{$\triangleright$ $u_0$ defined for convenience}
        \For{$i = 1,\ldots,n$}
            \State $u_i \gets u_{i-1} + g_i$
        \EndFor
        \State\Return $u$
    \end{algorithmic}

\caption{Computing matrix-vector product with the discrete derivative.}

\label{fig:fig1upd}

\end{figure}

\theom*

\begin{proof}
Let the input matrix $M$ be $d$-twin-ordered.
By \Cref{lem1}, there exists a rectangle-decomposition of $M$ with $\fO(dn)$ disjoint rectangles, which by \Cref{lem2}, implies that the discrete derivative $D$ of $M$ has support $\lvert \supp(D) \rvert \in \fO(dn)$, and can be computed in $\fO(n^2)$ time.
After this preprocessing, for each vector $v \in \mathbb{R}^n$, we can compute $Mv$ and $v^\intercal M$ in $\fO(n + \lvert \supp(D) \rvert) = \fO(dn)$ time using \Cref{lem3}.\qedhere
\end{proof}

\noindent\textit{Point update queries.}
Our framework can efficiently handle dynamic updates in the input matrix $M$. %
Consider an operation $\update(i, j, \beta)$ that sets $M_{i, j}$ to a value $\beta \in \{0,1\}$, where $(i, j) \in [n] \times [n]$. Update operations can be arbitrarily interleaved with matrix-vector product queries.

Such an update operation involves setting $M_{i, j} \gets \beta$ and updating $\mathcal{S}_D$ in
$\fO(1)$ time. Indeed, denoting by $\alpha$ the matrix entry $M_{i, j}$ before the update, only four entries of $D$ change, as follows:
\[
\begin{aligned}
D_{i,j} \gets D_{i,j}+(\beta - \alpha), \quad
D_{i+1,j} \gets D_{i+1,j}-(\beta - \alpha),\\
D_{i,j+1} \gets D_{i,j+1}-(\beta - \alpha), \quad
D_{i+1,j+1} \gets D_{i+1,j+1}+(\beta - \alpha).
\end{aligned}
\]

Note that initially, by Lemmas~\ref{lem1} and \ref{lem2}, $\lvert \supp(D) \rvert \in \fO(dn)$ and each update introduces at most four new entries to the support of $D$. Thus, after $r$ updates, $\lvert \supp(D) \rvert \in \fO(dn + r)$, so a matrix-vector multiplication takes $\fO\left(n + \lvert \supp(D) \rvert \right) = \fO(dn+r)$ time.

\section{Bounded twin-width and Hamming-coherence}
\label{sec4}

In this section, we study the Hamming-coherence of general bounded twin-width matrices, a connection that underpins our main result.

First, we need to formally define divisions and merge sequences of binary matrices.
Let $M \in \{0,1\}^{n \times n}$ be a binary matrix and let $\cR = (I_1, \dots, I_k)$ and $\cC = (J_1, \dots, J_{k'})$  be two partitions of $[n]$ into pairwise disjoint \emph{contiguous non-empty intervals}.
Then, we refer to $(\cR, \cC)$ as a \emph{division} of $M$; intuitively, the intervals in $\cR, \cC$ capture neighboring sets of rows and columns in $M$ that are \emph{merged} together.
By rows, resp., columns of the division $(\cR,\cC)$ we refer to the submatrices of the original matrix $M$ defined by a set of rows $I_i \in \cR$, resp., by a set of columns $J_j \in \cC$.
We call a submatrix of $M$ at the intersection of rows indexed by $I_i$ and columns indexed by $J_j$ the \emph{cell} $(i,j)$ of the division $(\cR,\cC)$.
A cell is \emph{constant}, if all its entries are equal.

A \emph{merge sequence} of~$M$ is a sequence of divisions $(\cR_{2n-1}, \cC_{2n-1}), \dots, (\cR_{1}, \cC_1)$ such that (i) $\cR_{2n-1} = \cC_{2n-1}$ are partitions of $[n]$ into singletons, (ii) $\cR_1 = \cC_1$ both consist of a single part $[n]$, and (iii) $(\cR_{i-1}, \cC_{i-1})$ is obtained from $(\cR_{i}, \cC_{i})$ by merging two adjacent intervals in either $\cR_i$ or $\cC_i$ (but not both).
A merge sequence is \emph{$d$-wide} if each row and column in every division $(\cR_i, \cC_i)$ contains at most $d$ non-constant cells.
Recall that $M$ is $d$-twin-ordered if it admits a $d$-wide merge sequence and it has twin-width at most~$d$ if there exists a reordering of its rows and columns such that the resulting matrix is $d$-twin-ordered.

Results for bounded twin-width matrices are usually stated for symmetric matrices, or rather for adjacency matrices of undirected graphs.
In this context, a symmetric variant of merge sequences is natural, which we define next.
A \emph{symmetric merge sequence} of a symmetric matrix $M \in \{0,1\}^{n \times n}$ is a sequence of divisions $(\cR_n, \cC_n), \dots, (\cR_{1}, \cC_1)$ such that (i) $\cR_n = \cC_n$ are partitions of $[n]$ into singletons, (ii) $\cR_1 = \cC_1$ both consist of a single part $[n]$, (iii) we have $\cC_i = \cR_i$ for each $i \in [n]$, and (iv) $(\cR_{i-1}, \cC_{i-1})$ is obtained from $(\cR_{i}, \cC_{i})$ by merging the same adjacent intervals in both $\cR_i$ and~$\cC_i$.
Every symmetric matrix of twin-width~$d$ admits a uniform reordering of its rows and columns (i.e., the same permutation
is applied to both, so the resulting matrix is again symmetric) such that the resulting matrix admits an $f(d)$-wide symmetric merge sequence~\cite{bonnet2021twin}.

Crucially, every symmetric matrix of bounded twin-width also admits a reordering of its rows and columns and a symmetric merge sequence in which, additionally, the height of every row and width of every column is controlled.
Formally, a $d$-wide symmetric merge sequence $(\cR_n, \cC_n), \dots, (\cR_{1}, \cC_1)$ is \emph{balanced} if we have $|I| \le d \cdot \frac{n}{i}$ for every interval $I \in \cR_i = \cC_i$.
The existence of balanced symmetric merge sequences for bounded twin-width symmetric matrices is implicit in~\cite{BonnetGKTW22}, see the discussion in~\cite[\S\,2.3]{BergeBDW23}.

\begin{lemma}[\cite{BonnetGKTW22}]\label{lem9}
For every symmetric binary matrix $M$ of twin-width~$d$, there exists a matrix~$M'$ obtained by a uniform reordering of its rows and columns such that $M'$ admits a $D$-wide balanced symmetric merge sequence where $D \in 2^{2^{\fO(d)}}$.
\end{lemma}

We now move to proving the main theorem of this section.
\thma*

\begin{proof}%
We only focus on the row-Hamming-coherence since twin-width is clearly preserved under transposition and $\cH^c(M) = \cH^r(M^\intercal)$.
Our goal is to utilize the existence of balanced merge sequences.
However, \cref{lem9} only applies to symmetric matrices.
This is resolved in a standard way by considering a symmetric $2n\times 2n$ matrix $M' = \begin{psmallmatrix} 0& M\\ M^\intercal &0\end{psmallmatrix}$ that has the same twin-width as~$M$. Indeed, given a $d$-wide merge sequence for $M$, simultaneously mirroring each step in both rows and columns of $M'$ yields a $d$-wide symmetric merge sequence: the all-zero diagonal blocks never introduce $\bot$ symbols, and the $M$ and $M^\intercal$ blocks each contribute at most $d$ non-constant cells per merged row or column by assumption.
Observe that an $\fO(2n \cdot \log 2n)$ bound on the row-Hamming-coherence of $M'$ directly translates to an $\fO(n \log n)$ bound on the row-Hamming-coherence of~$M$ since $\cH^r(M) \le \cH^r(M')$.
Thus, we may assume from now on that $M$ is a symmetric $n \times n$ binary matrix.

By \cref{lem9}, there is a matrix~$M'$ obtained from~$M$ by reordering its rows and columns that admits a $D$-wide balanced symmetric merge sequence $(\cR_n, \cC_n), \dots, (\cR_{1}, \cC_1)$ with $D \in 2^{2^{\fO(d)}}$.
As we already argued, reordering rows and columns bears no effect on the row-Hamming-coherence and thus, we have $\cH^r(M') = \cH^r(M)$.

The rest of the argument closely follows the proof of \cref{lem1} as given in~\cite{matrix_ds}.
Let $\cS_i$ be the set of all cells in~$(\cR_i, \cC_i)$, i.e., the set of all rectangles $I \times J$ for $I \in \cR_i$ and $J \in \cC_i$.
Set $\cS = \bigcup_{i=1}^n \cS_i$ and observe that $\cS$ is a laminar family, i.e., every two rectangles in~$\cS$ are either disjoint or one is contained in the other.

Let $\cQ$ be the subset of $\cS$ containing all 1-rectangles $R \in \cS$ that are inclusion-wise maximal within~$\cS$, i.e., there is no 1-rectangle $R' \in \cS$ such that $R \subsetneq R'$.
Observe that $\cQ$ covers all 1-entries of~$M'$ since every 1-entry is contained in some $1\times 1$ rectangle in~$\cS_n$ and moreover, all rectangles in~$\cQ$ are pairwise disjoint since $\cS$ is a laminar family.

Observe that the row-Hamming-sum of~$M'$ with respect to the canonical identity permutation is at most $\sum_{R \in \cQ}2w_R$ where $w_R$ denotes the width of a rectangle~$R$.
This holds since this canonical row-Hamming-sum is upper bounded by the length of all horizontal segments in the boundary of the orthogonal polygons of 1-entries in~$M'$, to which every rectangle in~$\cQ$ contributes at most twice its width.

For a 1-rectangle $R \in \cQ$, let $i_R$ denote the smallest~$i$ such that $R \in \cS_i$, i.e., the smallest~$i$ such that $R$ is a cell of the division $(\cR_i, \cC_i)$.
For each $i \in [n]$, let $\cQ_i$ be the subset of $\cQ$ containing all rectangles $R$ with $i_R = i$.
Naturally, $\cQ_1, \dots, \cQ_n$ form a partition of~$\cQ$.
Moreover, observe that we have $w_R \le D \cdot \frac{n}{i}$ for any $R \in \cQ_i$ due to the balancedness of the merge sequence.

We claim that each $\cQ_i$ contains at most $2D+2$ rectangles.
This clearly holds for $i = 1$ since $\cS_1$ itself contains only a single rectangle.
Fix $i$ with $2 \leq i \le n$ and consider an arbitrary $R \in \cQ_i$.
We know that $R \notin \cS_{i-1}$ and thus, $R$ is contained in a non-constant rectangle $T \in \cS_{i-1}$ in the unique row or column that appears in $(\cR_{i-1}, \cC_{i-1})$ for the first time.
Moreover, the rectangle $T$ is obtained from merging rectangle $R$ with a rectangle $R'$ that itself cannot be composed of only 1-entries since $T = R \cup R'$ is not a 1-rectangle.
The one exception is the rectangle lying in both the new row and the new column: it is the union of four rectangles of $\cS_i$ and can contain up to three rectangles of $\cQ_i$.
Since the merge sequence is $D$-wide, there are at most $2D$ such non-constant rectangles that appear in $\cS_{i-1}$ for the first time, at most $D$ per both the new row and column blocks.
It follows that $|\cQ_i| \le 2D+2$ for each $i \in [n]$.

We obtain the desired upper bound on $\cH^r(M')$ as
\[\sum_{R \in \cQ}2w_R = \sum_{i=1}^n\sum_{R \in \cQ_i}2w_R \le \sum_{i=1}^n (2D+2) \cdot 2D \cdot \frac{n}{i} \in \fO(D^2 \cdot n \log n) = 2^{2^{\fO(d)}} \cdot n \log n. \qedhere\]
\end{proof}

We complement our result by a lower bound, showing that the $n\log{n}$ factor in Theorem~\ref{thm3} is unavoidable, with proof in Appendix~\ref{appa}.
\thmb*

\section{$d$-mixed-free matrices}
\label{sec3}

In this section, we discuss $d$-mixed-free matrices. This class is, in some sense, of intermediate complexity between $d$-twin-ordered matrices and general, bounded twin-width matrices.

We say that a division $(\cR, \cC)$ of a binary matrix~$M$ is a \emph{$k$-division} if both $\cR$ and $\cC$ consist of exactly $k$ intervals.
Intuitively, a $k$-division splits the matrix into $k$ row and column blocks.
We say that a matrix~$M$ is \emph{$d$-grid-free} if there is no $d$-division such that all its cells contain some 1-entry.
Notably, $d$-grid-free matrices play an important role in extremal combinatorics with the landmark Marcus-Tardos theorem~\cite{MarcusTardos} showing that any $d$-grid-free $n\times n$ matrix contains at most $f(d) \cdot n$ 1-entries.

It follows that $d$-grid-free matrices are necessarily sparse and therefore, a different notion is needed to capture well-structured dense matrices.
This was a crucial observation at the birth of the twin-width theory.
We say that a matrix is \emph{mixed} if it is neither horizontal nor vertical, i.e., if neither all its rows nor all its columns are equal.
Equivalently, a matrix~$M$ is mixed if and only if it contains a \emph{corner}, i.e., a contiguous $2 \times 2$ submatrix that is neither horizontal nor vertical (see \cite{bonnet2021twin}).
A matrix $M$ is then \emph{$d$-mixed-free} if it does not admit any $d$-division with all its cells mixed.
Observe that any $d$-grid-free matrix is trivially $d$-mixed-free but not the other way around, as witnessed by, e.g., the all-ones matrix.
We now state the main result of this section, with proof in Appendix~\ref{appb}.

\lemb*

We complement our structural characterization of $d$-mixed-free matrices by observing that the exponential dependence on~$d$ is unavoidable for rectangle-decompositions of $d$-mixed-free matrices (proof in Appendix~\ref{appb}).

\lemc*

\section{Conclusions}
\label{sec5}

We gave efficient algorithms for matrix-vector multiplication, and consequently, matrix-matrix multiplication for matrices of bounded twin-width. Our approach is significantly simpler than earlier ones and requires only one of the matrices to have bounded twin-width.

In case of $d$-twin-ordered and $d$-mixed-free matrices, our runtime bounds do not contain spurious logarithmic factors, and have a mild dependence on the parameter $d$ -- linear, resp., single-exponential, improving the runtimes of earlier results, while being exceedingly simple.

In case of matrices of bounded twin-width where the ordering of rows and columns is arbitrary, our result improves standard matrix multiplication bounds, and is among the few algorithmic results that take advantage of bounded twin-width without explicitly using a contraction sequence that certifies it (see also~\cite{GroheN24}, in a different context, and~\cite{bonnet26,bonnet2026b,CMY26}). For simplicity, we stated all our results for square matrices; they apply, \emph{mutatis mutandis}, to rectangular matrices.

\begin{credits}
\subsubsection{\ackname}
Michal Opler was co-funded by the European Union under the project Robotics and advanced industrial production (reg.\ no.\ CZ.02.01.01/00/22\_008/0004590).
\end{credits}

 \bibliographystyle{splncs04}
 \bibliography{main}

\newpage

\appendix
\renewcommand{\theHsection}{\thesection.appendix}

\section{Proof of Theorem~\ref{thm4}}
\label{appa}

\thmb*

\begin{proof}%
We construct a binary $n \times n$ matrix $M = M_i$ where $n$ is of the form $n = 2 \cdot 3^i$, for any non-negative integer $i$.

We show that the twin-width of $M$ is at most $3$ and that both the row- and column-Hamming-coherence of $M$ are $\Omega(n\log{n})$. For values of $n$ not of the given form, we can construct $n \times n$ matrices with the same properties by simple padding with $0$s.

The construction is recursive. We take $M_0 = \begin{psmallmatrix} 1& 0\\ 0 &1\end{psmallmatrix}$ for the base case.

For $i>0$, we let $M_i = \begin{psmallmatrix} M_{i-1}& 0 & 1\\ 1 &M_{i-1}&0\\ 0 & 1 & M_{i-1}\end{psmallmatrix}$,
where the $0$ and $1$ entries denote blocks of $0$-entries, resp., $1$-entries, of the same $(2 \cdot 3^{i-1} \times 2 \cdot 3^{i-1})$ sizes as the copies of $M_{i-1}$.

\paragraph*{Twin-width.} We first show the bound on the twin-width; in fact, we more strongly show that $M_i$ is $3$-twin-ordered. We proceed by induction on $i$. The matrix $M_0$ is clearly $2$-twin-ordered.

For $M_i$, when $i>0$, we exhibit a merge sequence (see definition in \S\,\ref{sec1}). In the first phase, we only allow merges among the first $2 \cdot 3^{i-1}$ rows and the first $2 \cdot 3^{i-1}$ columns. Notice that among the 9 blocks of $M_i$, only five are affected (the three on top and the three leftmost, with the top left $M_{i-1}$ block shared). Since $M_{i-1}$ is $3$-twin-ordered, by the induction hypothesis, we can perform these merges in such a way that at any time at most $3$ error symbols $\bot$ are present in any row or column of the top left block. Notice that in the $0$- and $1$-blocks we cannot get a symbol $\bot$, so this upper bound of $3$ holds globally.

After this phase, we end up with the matrix $\begin{psmallmatrix} \bot & 0 & 1\\ 1 &M_{i-1}&0\\ 0 & 1 & M_{i-1}\end{psmallmatrix}$, where the $0$ and $1$ entries in the top row, resp., leftmost column indicate a size-$(2\cdot3^{i-1})$ row-, resp., column-vector block of equal elements.

Analogously, in the second and third phases we contract the rows and columns corresponding to the middle $M_{i-1}$ block, resp., the lower right $M_{i-1}$ block. Throughout the process the number of $\bot$ symbols in any row or column is at most $3$, and we end up with the $3 \times 3$ matrix
$\begin{psmallmatrix} \bot& 0 & 1\\ 1 &\bot&0\\ 0 & 1 & \bot\end{psmallmatrix}$.
The merge sequence can then be finished in an arbitrary way since matrices of size at most $3 \times 3$ cannot contain more than three $\bot$ symbols in any row or column.
This certifies that $M_i$ is $3$-twin-ordered.

\paragraph*{Hamming-coherence.} We next show the lower bound on the row-Hamming-coherence $\cH^r(M_i)$. The argument for the column-Hamming-coherence is entirely symmetric, we thus omit it.

Consider the rows $r^{(1)}, \dots, r^{(n)}$ of $M = M_i$ and let $\pi$ be an ordering of them of minimal row-Hamming-sum, i.e., $\cH_\pi^r(M) = \cH^r(M)$.

Call a row \emph{type}-$(1)$, resp., type-$(2)$ or type-$(3)$, if it is among the first (resp., second or third) $2\cdot3^{i-1}$ rows of $M$.

We first observe that every row of $M$ has exactly the same number of $0$ and $1$ entries. It follows, by simple inspection of $M_i$, that the Hamming distance between two rows of different types is exactly $2n/3$. (They differ in exactly half of the entries that fall into an $M_{i-1}$-block in one of the two rows, and differ in all remaining entries). On the other hand, two rows of the same type have Hamming distance at most $n/3$. (They can only differ in the entries that fall into the $M_{i-1}$-block.)

From these observations we can conclude that in the optimal ordering $\pi$ all rows of the same type appear as contiguous subsequences.

Suppose otherwise that there are two rows, say $r^{(\pi(i))}$ and $r^{(\pi(j))}$ that are of the same type, say type $(1)$, with $i<j$, but $r^{(\pi(i+1))}$ and $r^{(\pi(j-1))}$ are both of a type other than $(1)$ (these two rows may be the same).

Then we can modify $\pi$ by cutting out the contiguous subsequence $\pi(i+1), \dots, \pi(j-1)$ and placing it at the end of the sequence. Observe that removing the subsequence reduces the total Hamming-sum by $2 \cdot 2n/3$. Placing $\pi(i)$ next to $\pi(j)$ adds at most $n/3$ to the sum, and placing the sequence to the end adds at most $2n/3$. The total row-Hamming-sum has decreased and we reduced the number of neighboring pairs of different types. Repeating this procedure, we end up with rows of all three types forming contiguous subsequences in $\pi$. %

It remains to compute the total Hamming distance $\cH_\pi^r(M_i)$ for this ordering.
Denote this quantity by $\cH(i)$. Within the contiguous same-type subsequences, only the $M_{i-1}$-block entries contribute, and we can analyze them recursively. At the two interfaces between types we have a Hamming distance of exactly $2n/3$ each.
We thus have \[\cH(i) = 3 \cdot \cH(i-1) + 4n/3 = 3 \cdot \cH(i-1) + 8\cdot 3^{i-1},\] with base case $\cH(0) = 2$. The recurrence yields $\cH^r(M) = \cH(i) \geq  3^{i-1} \cdot (8i + 6) \in \Omega(n\log{n})$. \qedhere
\end{proof}

\section{Proofs for $d$-mixed-free matrices}
\label{appb}

Recall that every $d$-twin-ordered matrix is also $(2d+2)$-mixed-free~\cite{bonnet2021twin}, which motivates the use of \emph{mixed-freeness} as an ``intermediate'' complexity measure. Note, however, that bounded mixed-freeness does not imply bounded twin-order.
We illustrate this gap in complexity with a simple example\footnote{We thank Marek Soko\l{}owski for communicating this example.} that shows an $n \times n$ matrix $M$ that is $2$-mixed-free but not $d$-twin-ordered for any $d < n/2$.

Assume $n$ is even and let $Q$ be an $n/2 \times n/2$ matrix whose odd-indexed rows are all $0$ and even-indexed rows are all $1$. Then, let $M = \begin{psmallmatrix}
    Q^\intercal &Q\\Q&Q^\intercal
\end{psmallmatrix}$.  The claims on the complexity of $M$ are easily verified. For concreteness, we show the construction for $n=8$:
\[M = \begin{psmallmatrix}
    0 & 1 & 0 & 1 & 0 &0 &0 &0\\
    0 & 1 & 0 & 1 & 1 &1 &1 &1\\
    0 & 1 & 0 & 1 & 0 &0 &0 &0\\
    0 & 1 & 0 & 1 & 1 &1 &1 &1\\
    0 & 0 & 0 & 0 & 0 &1 &0 &1\\
    1 & 1 & 1 & 1 & 0 &1 &0 &1\\
    0 & 0 & 0 & 0 & 0 &1 &0 &1\\
    1 & 1 & 1 & 1 & 0 &1 &0 &1
  \end{psmallmatrix} .\]

\lemb*
\begin{proof}%
It was shown in~\cite[Lemma 8]{matrix_ds} that every $d$-mixed-free matrix contains at most $2^{\fO(d)}\cdot n$ corners.
Let $\cP$ denote the collection of disjoint orthogonal polygons covering all 1-entries in~$M$.
Recall that a single polygon~$P \in \cP$ always admits a rectangle-decomposition of size at most $C_P - H_P + 1$ where $C_P$ is the number of its concave vertices and $H_P$ is the number of its holes.
We ignore the terms due to the holes and show that the sum $\sum_{P \in \cP}(C_P+1)$ over all polygons can be suitably charged to the $2^{\fO(d)}\cdot n$ corners present in~$M$.

Observe that any concave vertex corresponds uniquely to the mixed submatrix $\begin{psmallmatrix}
    1 &0\\1&1
\end{psmallmatrix}$ or its three possible rotations and thus, the number of concave vertices in all polygons is upper bounded by the total number of corners in~$M$.
It remains to bound the total number of polygons in~$\cP$.
For this purpose, associate to each polygon $P \in \cP$ an arbitrary leftmost convex vertex of $P$.
Observe that this point is unique to each polygon and it either lies on the boundary of $M$, or it uniquely corresponds to either one of the mixed matrices $\begin{psmallmatrix}
    0 &1\\0&0
\end{psmallmatrix}$,
$\begin{psmallmatrix}
    1 &0\\0&1
\end{psmallmatrix}$ or
their reflections along the horizontal axis.
Therefore, the total number of polygons is also upper bounded by the number of corners contained in~$M$, plus a linear term in $n$ for the polygons intersecting the boundary of~$M$.
We see that $M$ admits a rectangle-decomposition of size at most
\[\sum_{P \in \cP} (C_P - H_P + 1) \le \sum_{P \in \cP} C_P + \sum_{P \in \cP} 1 \le 2^{\fO(d)}\cdot n + 2^{\fO(d)}\cdot n + \fO(n) \in 2^{\fO(d)} \cdot n.\]
\qedhere
\end{proof}

We complement our structural characterization of $d$-mixed-free matrices by observing that the exponential dependence on~$d$ is unavoidable for rectangle-decompositions of $d$-mixed-free matrices.

\lemc*
\begin{proof}%
Fox~\cite{jfox} showed that for every positive~$d$ and~$n$, there exists a $d$-grid-free matrix~$M \in \{0,1\}^{n \times n}$ with $2^{\Omega(d^{1/2})} \cdot n$ 1-entries.
We define a $2n \times 2n$ binary matrix $M'$ where $M'_{2i,2j} = M_{i,j}$ for each $i, j \in [n]$ and all other entries in~$M'$ are $0$.
Clearly, $M'$ is still $d$-grid-free and, thus, also $d$-mixed-free.
Moreover, $M'$ admits a unique rectangle-decomposition that consists of $1 \times 1$ rectangles covering all its 1-entries since no pair of 1-entries is adjacent.
Therefore, $M'$ is a $d$-mixed-free $2n \times 2n$ binary matrix that requires a rectangle-decomposition of size at least $2^{\Omega(d^{1/2})} \cdot n$. \qedhere
\end{proof}

\end{document}